\documentclass[letterpaper, 10 pt, conference]{ieeeconf} 
\IEEEoverridecommandlockouts                              

\usepackage{graphics,epstopdf,wrapfig}
\usepackage{epsfig}
\usepackage{amsmath}
\usepackage{mathrsfs}
\usepackage{epsfig}
\usepackage{amssymb,latexsym,amsfonts,amsmath}
\usepackage{graphicx}
\usepackage{subfigure}
\usepackage{color}
\usepackage{epstopdf}

\usepackage{amsthm}
\usepackage{float}

\newtheorem{thm}{Theorem}
\newtheorem{lemma}{Lemma}
\newtheorem{problem}{Problem}
\newtheorem{definition}{Definition}
\usepackage{tikz}
\usetikzlibrary{automata,positioning,shapes,arrows}
\usetikzlibrary{shapes,snakes}
\tikzstyle{block} = [draw, rectangle, minimum size=2em]
\theoremstyle{remark}
\newtheorem{remark}{Remark}
\newtheorem{assumption}{Assumption}
\tikzset{>=latex}

\theoremstyle{exampstyle} 

\author{Bo~Wu, Zhiyu~Liu and Hai~Lin
	\thanks{The partial support of the National Science Foundation (Grant No. CNS-1446288, ECCS-1253488, IIS-1724070) and of the Army Research Laboratory (Grant No. W911NF- 17-1-0072) is gratefully acknowledged.}
	\thanks{ Bo Wu, Zhiyu Liu and Hai Lin are with the Department of Electrical Engineering, University of Notre Dame, Notre Dame,
		IN, 46556 USA. {\tt\small bwu3@nd.edu, zliu9@nd.edu, hlin1@nd.edu}}}

\begin{document}
	
	\title{\Large \bf Parameter and Insertion Function Co-synthesis for Opacity Enhancement  in Parametric Stochastic Discrete Event Systems}

	\maketitle
	
	\begin{abstract}
		Opacity is a property that characterizes the system's capability to keep its ``secret'' from being inferred by an intruder that partially observes the system's behavior. In this paper, we are concerned with enhancing the opacity using insertion functions, while at the same time,  enforcing the task specification in a parametric stochastic discrete event system. We first obtain the parametric Markov decision process that encodes all the possible insertions. Based on which, we convert this parameter and insertion function co-synthesis problem into a nonlinear program. We prove that if the output of this program satisfies all the constraints, it will be a valid solution to our problem. Therefore, the security and the capability of enforcing the task specification can be simultaneously guaranteed.
	\end{abstract}
	
	\IEEEpeerreviewmaketitle
	
	\section{Introduction}
	Opacity refers to a system's capability to hide its ``secret" information from being inferred by an outside intruder. It is often assumed that the intruder knows the system dynamic or even the opacity enforcing strategy, but with only partial observability to the system behavior. The secret is said to be \emph{opaque} to the intruder if, for every system behavior relevant to the secret, there is an observationally equivalent non-secret behavior, therefore the intruder is never sure about whether the secret has occurred or not.
	
	The research on opacity has received increasing interest because of its applications in cybersecurity. There has been abundant existing work in the last decade. The opacity problem was first introduced in the computer science community \cite{mazare2004using} and quickly spread to discrete event system (DES) researchers, see e.g \cite{saboori2007notions,bryans2008opacity,wu2014synthesis}. 
	
	There are essentially two main directions in the opacity research --- verification and enforcement. The verification problem studies whether the system is opaque or not given the system model, observation mask, and the secret. There are many notions of opacity in the existing literature such as language-based opacity (LBO) (the secret is in the form of regular or $\omega$ regular languages), initial-state opacity (ISO) (the initial state is the secret), current-state opacity (CSO) (the secret is revealed if the intruder is sure that the current state is a secret state), initial-and-final-state opacity (IFO) (the initial and final states are a secret pair) and $K$-step opacity (the intruder cannot infer any secret behavior happened in the past $K$ steps).  A verification algorithm for LBO was introduced in \cite{lin2011opacity} and its relationship with observability and diagnosability of DES was discussed. The verification of ISO and IFO was studied in \cite{wu2013comparative}. Cassez et.al \cite{cassez2009dynamic} explores the verification of CSO. It has been shown that these four notions of opacity are equivalent and can be transformed from one to another \cite{wu2013comparative}. Therefore, in this paper, we focus on CSO. Furthermore, the verification of $K$-step and infinite step (where $K\rightarrow\infty$) opacity were studied in \cite{saboori2011verification,saboori2012verification}.  In recent years, opacity was also extended to probabilistic systems to provide a quantitative measure of opacity instead of just a yes or no binary answer. CSO in probabilistic finite automata was introduced in \cite{saboori2014current}. In \cite{berard2015probabilistic}, quantification of LBO in terms of $\omega$-regular languages for Markov Decision Process was studied and the problem was transformed to  probabilistic model checking \cite{baier2008principles}. 
	
	On the other hand, there is also much progress in opacity enforcement --- synthesizing functions that modify observed system behavior such that the opacity can be enforced or maximized. Supervisory control theory \cite{ramadge1987supervisory} was adopted in opacity-enforcing in \cite{saboori2008opacity,saboori2012opacity,darondeau2015enforcing,yin2016uniform} where the supervisor dynamically disables certain system behaviors that would reveal the secret. Dynamic observer approach was proposed in \cite{cassez2012synthesis} where the observability of every system event was dynamically changed. However, the approaches mentioned above either constrain the full system behavior or may create new observed behaviors that don't exist in the original system which leaves the clue of the defense model. As a result, insertion functions that dynamically insert observable events \cite{wu2014synthesis} and more recently, edit functions \cite{wu2016obfuscator}, which can also erase events were introduced. The basic idea is to introduce a game structure (called all insertion/edit structure) that encodes all the valid system and insertion/edit function moves. Then the synthesis of opacity enforcing functions is equivalent to finding a winning strategy such that no matter what the original system outputs, the opacity can be enforced. The stochastic extension to insertion function synthesis for maximized opacity was then introduced in \cite{wu2016enhancing} where a Markov Decision Process (MDP) similar to AIS  was constructed, and dynamic programming was applied to find the optimal insertion function. 
	
	To the best of our knowledge, however, there are no results on opacity enforcement in parametric models while also considering other task specifications. Therefore, in this paper, we are motivated to fill this gap. By considering a parametric stochastic discrete event system (PSDES) model and a given task specification, we first get the parametric Markov Decision Process (PMDP) that encodes all the possible insertion actions. Then the insertion function synthesis and task specification enforcement problem is converted into a nonlinear programming (NLP) that can be solved to simultaneously synthesize the insertion function and the parameters. We prove the correctness of our NLP if it finds a valid solution that respects all our constraints. 
	
	In the rest of this paper, Section \ref{sec:pre} defines the relevant models. Section \ref{sec:opacity} introduces the basic of the opacity notion and the insertion mechanism. Section \ref{sec:prob} formulates our parameter and insertion function co-synthesis problem. Section \ref{sec:main_results} presents our main results while running through a motivating example. Section \ref{sec:conclusion} concludes the paper.

	\section{Preliminaries}\label{sec:pre}
	In this paper, we consider the opacity in the framework of the parametric stochastic discrete event system (PSDES). 
	\subsection{Parametric Stochastic Discrete Event System (PSDES)}
	We first introduce discrete event systems (DES)  modelled as non-deterministic finite-state automata (NFA) \cite{cassandras2009introduction}  $G=(Q,\Sigma,\delta,Q_0)$, where $Q$ is a finite set of states, $\Sigma$ is a finite set of events, $\delta:Q\times \Sigma^*\rightarrow 2^Q$ is a transition function, $Q_0\subseteq Q$ is a set of initial states. The generated language $\mathcal{L}(G)=\{\omega\in \Sigma^*|\exists q_0\in Q_0, \delta(q_0,\omega) \text{ is defined}\}$. $G$ is assumed to be partially observable and $\Sigma$ is partitioned into two disjoint sets, namely observable set $\Sigma_o$ and unobservable set $\Sigma_{uo}$ such that $\Sigma_o\cup\Sigma_{uo}=\Sigma$. Given a string $\omega\in\Sigma^*$, an observation mask (natural projection) $O:\Sigma^*\rightarrow\Sigma_{o}^*$ is defined recursively as $O(\omega)=O(\omega')O(\sigma)$ where $\omega=\omega'\sigma$, $\omega'\in\Sigma^*$ and $\sigma\in\Sigma$, $O(\sigma)=\sigma$ if $\sigma\in\Sigma_{o}$ and $O(\sigma)=\epsilon$ if $\sigma\in\Sigma_{uo}\cup\{\epsilon\}$, where $\epsilon$ stands for the empty string. Given strings $\omega,\omega'$, if $\omega'$ is a prefix of $\omega$, we denote it as  $\omega'\preceq\omega$. If $\omega'$ is a strict prefix of $\omega$, we denote it as $\omega'\prec\omega$.
	
	A Stochastic Discrete Event System (SDES)  is denoted by $H=(Q,\Sigma,P,\pi_0)$, where $Q$ is a finite set of states, $\Sigma$ is a finite set of events, $P(q,\sigma,q')\rightarrow [0,1]$ is a transition function specifying the probability to transit from $q\in Q$ to $q'\in Q$ with the event $\sigma\in\Sigma$, $\pi_0:Q\rightarrow[0,1]$ defines the initial distribution. 
	
	In this paper, we have the following assumption to the SDES model.
	\begin{assumption}\label{asp:wu}
		Given an SDES $H=(Q,\Sigma,P,\pi_0)$,  we assume that $\sum_{q'\in Q}\sum_{\sigma\in\Sigma}P(q,\sigma,q')\in\{0,1\}$. We also require that $\sum_{q'\in Q}\sum_{\sigma\in\Sigma_{uo}}P(q,\sigma,q')<1$.    
	\end{assumption}
	This assumption essentially requires that each state is either a sink with no outgoing transition, or all the transitions sum up to $1$. If there is an unobservable transition from
	$s$ to $s'$, there must also be at least one observable
	transition from $s$ to $s'$. It can be observed that we can associate an SDES $H=(Q,\Sigma,P,\pi_0)$ to an NFA $G=(Q,\Sigma,\delta,Q_0)$, where $\delta(q,\sigma,q')$ is defined if and only if $P(q,\sigma,q')>0$. $Q_0=\{q|\pi_0(q)>0\}$. Therefore, the generated language of $H$ is defined to be, $\mathcal{L}(H)=\mathcal{L}(G)$.  
	Let $|Q|=n$, i.e., there are $n$ states in $H$, we denote $P_\sigma$ where $\sigma\in\Sigma$ as an $n\times n$ matrix where $P_\sigma(i,j)=P(q_i,\sigma,q_j)$. $P_\sigma$ can be naturally extended to $P_\omega$ where $\omega\in \Sigma^*$, $P_\omega(i,j)=P(q_i,\omega,q_j)$. For $\omega=\sigma_1\sigma_2...\sigma_k\in \Sigma^*$, $P_\omega=\prod_{i=1}^{k}P_{\sigma_i}$.
	
	A PSDES $PH=(Q,\Sigma,P,\pi_0,V)$ is an SDES that satisfies Assumption \ref{asp:wu}, where  its transition function $P(q,\sigma,q')=f_{q,\sigma,q'}(V)$ where $V=\{v_1,v_2,...,v_m\}$ is a finite set of  parameters that are strictly positive and real-valued, and $f_V\in \mathcal{F}_V$. In this paper, $\mathcal{F}_V$ is a set of posynomial functions \cite{boyd2007tutorial} which are defined in the form of
	\begin{equation}\label{eqn:pos}
		f_V = \sum_{k=1}^{K}c_kv_1^{a_{1k}}...v_m^{a_{mk}}   
	\end{equation}
	where $c_k\in\mathbb{R}_{>0},a_{ik}\in\mathbb{R}$. Therefore, the transition probabilities are parameterized by $V$. The transition probabilities of the parametric models in many existing benchmarks can be converted to this class \cite{cubuktepe2017sequential}. A valuation $v(V)\in\mathbb{R}^m$ maps $V$ to $\mathbb{R}^m$. $Val^v$ denotes the set of all valuations.
	
	Figure \ref{fig:PSDES} shows a communication network modelled as a PSDES $PH=(Q,\Sigma,P,\pi_0,V)$, where $Q=\{0,1,2,...,10\}$, $\Sigma=\{a,b,c\}$, $\pi_0(0)=1$, $V=\{v_1,...,v_7\}$. The information generated in communication node $0$ must be transmitted to node $10$ through routing. $\Sigma$ denotes the message types and $P(q,\sigma,q')$ denotes the probability that the communication node $q$ decides to transmit the data to $q'$ for routing purpose. The label on each transition arrow from $q$ to $q'$ is in the form of $\sigma,p_v$ from $P(q,\sigma,q')=p_v$. We omit the transition probability if it is $1$.
	
	\begin{figure}
		\centering	
		\begin{tikzpicture}[shorten >=1pt,node distance=2cm,on grid,auto, bend angle=20, thick,scale=0.7, every node/.style={transform shape}] 
		\node[state,initial] (s0)   {$0$}; 
		\node[state] (s1) [above right= of s0] {$1$}; 
		\node[state] (s2) [right= of s1]  {$2$}; 
		\node[state] (s3) [right= of s2] {$3$}; 
		\node[state] (s4) [right= of s0] {$4$}; 
		\node[state] (s5) [right= of s4] {$5$}; 
		\node[state] (s6) [right= of s5]  {$6$}; 
		\node[state] (s7) [below right= of s0]  {$7$}; 
		\node[state,shade] (s8) [right= of s7]  {$8$}; 
		\node[state,shade] (s9) [below right= of s7]  {$9$}; 
		\node[state] (s10) [right= of s6]  {$10$}; 
		\path[->]

		(s0) edge node {$a,v_1$} (s1) 
		(s1) edge node {$b$} (s2) 
		(s2) edge node {$c$} (s3) 
		(s3) edge node {$b$} (s10) 
		(s0) edge node {$c,v_2$} (s4)  
		(s4) edge node {$a,v_4$} (s2) 
		(s4) edge node {$b,v_5$} (s5)  
		(s5) edge node {$a$} (s6)  
		(s6) edge node {$b$} (s10)  
		(s0) edge node [below left] {$b,v_3$} (s7) 
		(s7) edge node {$a,v_6$} (s8)
		(s7) edge node [below left] {$c,v_7$} (s9)
		(s8) edge node [below left] {$b$} (s10)
		(s9) edge node [below left] {$b$} (s10)
		
		; 
		\end{tikzpicture} 
		\caption{An communication network modelled as PSDES}\label{fig:PSDES}
	\end{figure}
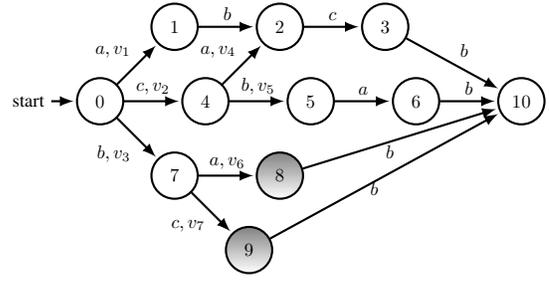

	\subsection{Parametric Markov Decision Process (PMDP)}
	An MDP is a tuple $\mathcal{M}=(S,\pi_0,A,T)$ where $S=\{s_0,s_1,...\}$ is a finite set of states, $\pi_0:S\rightarrow[0,1]$ is the initial distribution, $A$ is a finite set of actions, $T(s,a,s'):=Pr(s'|s,a)$, for $s,s'\in S,a\in A$.
	
	For each state $s\in S$, we denote $A(s)$ as the set of available actions. A Discrete Time Markov Chain (DTMC) is a special case of MDP with $|A(s)|<=1$ for all $s\in S$, where $|A(s)|$ is the cardinality of the set $A(s)$. It can be observed that SDES $H=(Q,\Sigma,P,\pi_0)$ can be seen as an MDP (more precisely, a DTMC) $\mathcal{M}=(Q,\pi_0,A,T)$, where $|A|=1$ and thus can be arbitrarily defined,  $T(q,a,q')=\sum_{\sigma\in\Sigma}P(q,\sigma, q')$. 
	
	Reasoning on an MDP requires resolving its nondeterminism in the action selection, which is done by a scheduler. Formally, a (memoryless) scheduler of a given MDP $\mathcal{M}$ is defined to be $\mu(s,a)\in [0,1]$ which denotes the probability of choosing $a\in A(s)$ at a state $s\in S$. 
	
	A parametric MDP is a tuple $\mathcal{M}=(S,\pi_0,A,T,V)$, where $V=\{x_1,...,x_n\}$ is a finite set of parameters and $T$ is in the form $T:S\times A\times S\rightarrow \mathcal{F}_V$, where $\mathcal{F}_V$ is the set of posynomial functions.

	\section{Opacity notion and enforcement} \label{sec:opacity}
	\subsection{Current state opacity}
	In this paper, we focus on the current state opacity (CSO) property. 
	\begin{definition} Given a DES $G=(Q,\Sigma,\delta,Q_0)$, observation mask $O$, and a set of secret states $Q_s\subseteq Q$, the system is current-state opaque if $\forall \omega\in \mathcal{L}_s(G)=\{\omega\in\mathcal{L}(G)|\exists q_0\in Q_0,\delta(q_0,\omega)\cap Q_s\neq\emptyset\}$, there exists another string $\omega'$, such that $\omega'\in \mathcal{L}_{ns}(G)=\{\omega\in\mathcal{L}(G)|\exists q_0\in Q_0,\delta(q_0,\omega)\cap (Q\backslash Q_s)\neq\emptyset\}$ and $O(\omega)=O(\omega')$.
	\end{definition}
	It can be seen that $\mathcal{L}_{s}(G)$ and $\mathcal{L}_{ns}(G)$ are not necessarily disjoint, that is, there may exist $\omega\in\mathcal{L}_{s}(G)\cap\mathcal{L}_{ns}(G)$. By definition, such string will not violate the CSO requirement. Intuitively, CSO requires that, if there is a string $\omega$, such that $\exists q_0\in Q_0$, $\delta(q_0,\omega)\subseteq Q_s$, that is, the system state, after executing $\omega$, lands in some secret state for sure, then there must exist another string $\omega'\in \mathcal{L}_{ns}(G)$, such that $O(\omega)=O(\omega')$. That is, they have the same observation but $\omega'$ may take the system to some nonsecret states. Therefore, if there is an intruder that knows $G$ and $Q_s$ and  observes the events with the observation mask $O$, it will never be sure if the system is currently in some secret states. Thus, whether the current system state is a secret state remains opaque to the intruder. Formally, the set of observable strings that will never reveal the secret can be written as $\mathcal{L}_{safe} = O[\mathcal{L}_{ns}(G)]\backslash((O[\mathcal{L}(G)]\backslash O[\mathcal{L}_{ns}(G)])\Sigma^*_o)$.
	
	Verifying CSO can be done by constructing its observer automaton in a standard way as in \cite{cassandras2009introduction}, and then checking if any observer state contains solely secret states. Take the system in Figure \ref{fig:PSDES} as an example. It can be seen as an NFA if we ignore the probabilities. The shaded states, i.e., the communication nodes $8$ and $9$ are vulnerable to potential attacks. If there is an outside intruder finds out that the message reaches either $8$ and $9$, the message may be intercepted. That is $Q_s=\{8,9\}$. We assume that all the events are observable. Its observer automaton is identical to the original system.  It can be found that this system is not opaque because if the intruder observes $ba$ or $bc$, it will be sure that the system is currently in state $8$ or $9$.
	
	Given an SDES $H$, it is possible to quantify the level of CSO \cite{wu2016enhancing}. We denote the set $\mathcal{L}_{rs}=\{\omega\in\mathcal{L}(H)|O(\omega)\notin\mathcal{L}_{safe}, \forall \omega'\prec\omega, O(\omega')\in\mathcal{L}_{safe}\}$ as all the strings that will reveal the secret for the first time. Then the opacity level, or the probability that the secret is never revealed can be computed as $P_{CSO}=1-P(\mathcal{L}_{rs})$. The computation details can be found in \cite{wu2016synthesis}.
	
	\subsection{Insertion function}
	When the opacity property does not hold, it is desired to design mechanisms to enforce (in DES) or enhance it (in SDES). In this paper, we use the insertion functions that inserts extra observable events before each system event and then output the modified string. From the intruder's perspective, the inserted observable events are not distinguishable from the observable events that actually happened. Moreover, we assume that the intruder does not have any information about the structure of the insertion function. In this way, we can modify a string $\omega$ into $\omega'$ such that $|O(\omega)|<=|O(\omega')|$, to trick the intruder to think that the system lands in other states even if the real system lands in a secret state.
	
	The insertion function $I:\Sigma_o^*\times \Sigma_o\rightarrow\Sigma_o^*\Sigma_o$ inserts a string of observable events before a observable system event based on the history of previous observable system events. For example, if there is an observable string $\omega\sigma$ where $\omega\in\Sigma_o^*,\sigma\in\Sigma_o$, then $I(\omega,\sigma)=\omega'\sigma$ where $\omega'\in\Sigma_o^*$ is the inserted string. In Figure \ref{fig:PSDES}, from the initial state $0$, suppose $b$ happens, the state will transit to $7$. But if we insert an event $a$ before $b$, that is $I(b)=ab$, then the intruder will observe $ab$ and thus thinks the current state is $2$.  With a slight abuse of the notation, the insertion function can be naturally extended to $I(\epsilon)=\epsilon$, $I(\omega\sigma)=I(\omega)I(\omega,\sigma)$. Then the modified language given a DES $G$ can be written into $I(O[\mathcal{L}(G)])=\{\omega'\in\Sigma_o^*|\exists \omega\in O[\mathcal{L}(G)], I(\omega)=\omega' \}$.
	
	The insertion function $I$ should satisfy private enforceability \cite{wu2014synthesis}. Namely, it should be defined to all the possible observable behaviors of the system. Formally, $\forall \omega\sigma\in O[\mathcal{L}(G)]$ where $\omega\in\Sigma_o^*,\sigma\in\Sigma_o$, $\exists \omega'\in\Sigma_o^*$, such that $I(\omega,\sigma)=\omega'\sigma$. Furthermore, the output of the insertion function should be privately safe by enforcing CSO. Formally, $I(O[\mathcal{L}(G)])\in \mathcal{L}_{safe}$. That is, after insertion, the observed behavior by the intruder should always lie in $\mathcal{L}_{safe}$.
	
	With an insertion function $I$ in an SDES $H$, the requirement to absolute private enforceability may be relaxed. That is, there may not exist a well-defined insertion function that achieves  CSO enforcement with probability $1$. To compare different insertion functions, we can compute the opacity level after the insertion. Similarly, we denote the set $\mathcal{L}^I_{rs}=\{\omega\in\mathcal{L}(H)|I[O(\omega)]\notin\mathcal{L}_{safe}, \forall \omega'\prec I[O(\omega)], \omega'\in\mathcal{L}_{safe}\}$ as all the strings that will reveal the secret for the first time after insertion. Then the opacity level can be computed as $P^I_{CSO}=1-P(\mathcal{L}^I_{rs})$. 
	
	\section{Problem Formulation}\label{sec:prob}
	Given a PSDES $PH=(Q,\Sigma,P,\pi_0,V)$, we would like to synthesize the parameters $V$ and the insertion function $I$, such that the opacity level $P_{CSO}$ can be no less than a given threshold $\gamma$. Furthermore, the PSDES should also satisfy certain task property. In this paper, we are  interested in the reachability specifications written as $\phi_t$ where 
	$$
	\phi_t = P_{\leq \lambda}(\diamondsuit D)
	$$
	where $0\leq\lambda\leq 1$, $D\subseteq Q$ and $\diamondsuit$ denotes ``eventually''. That is, we require that the probability to reach any undesired state $q\in D$ is bounded by $\lambda$. This reachability probability can be computed by making the states in $D$ absorbing£¬ i.e., introducing self loop with probability one on these states, and treat the PSDES as a Parametric DTMC where the specification $\phi_t$ can be efficiently verified \cite{quatmann2016parameter}. Formally,
	
	\begin{problem}
		Given a PSDES $PH=(Q,\Sigma,P,\pi_0,V), 0\leq \gamma\leq 1$, observation mask $O$, and the specification $\phi_t = P_{\leq \lambda}(\diamondsuit D)$, suppose $P_{CSO}<\gamma$, find a valuation $v\in Val^V$ and an insertion function $I$ such that $P^I_{CSO}\geq\gamma$ and $\phi_t$ is satisfied. 
	\end{problem}
	
	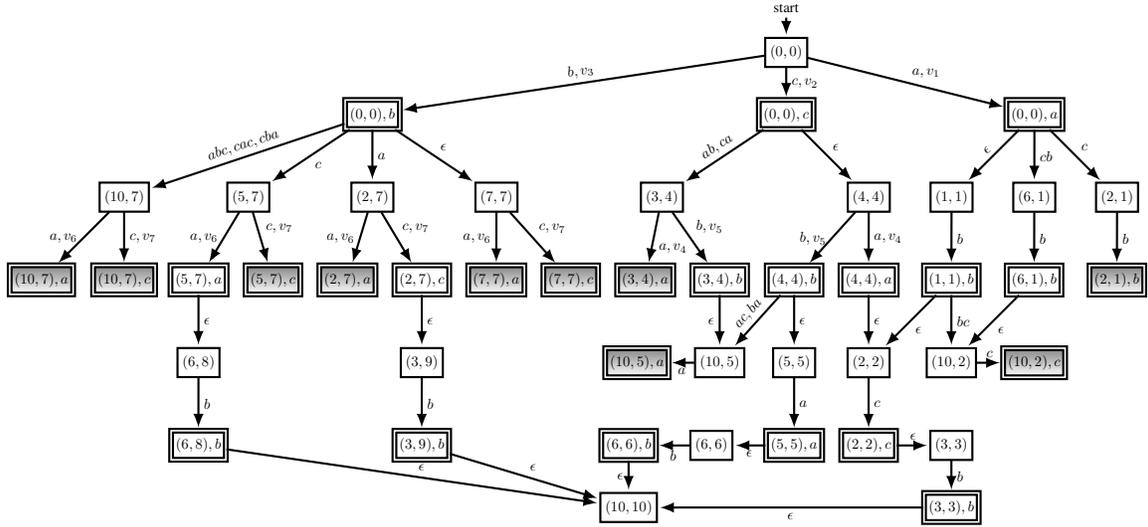
\begin{figure*}[ht]
		\centering
		\begin{tikzpicture}[shorten >=1pt,node distance=2cm,on grid,auto, bend angle=20, thick,scale=0.55, every node/.style={transform shape}] 
		\node[block,initial above] (s0)   {$(0,0)$}; 
		\node[double,block] (s1) [below left = 1.5 cm and 10 cm of s0] {$(0,0),b$}; 	
		\node[double,block] (s2) [below = 1.5 cm of s0]  {$(0,0),c$}; 
		\node[double,block] (s3) [below right= 1.5 cm and 6 cm of s0] {$(0,0),a$}; 
		
		\node[block] (s4) [below left= 2 cm and 6 cm of s1] {$(10,7)$}; 
		\node[block] (s5) [right = 3 cm  of s4] {$(5,7)$}; 
		\node[block] (s6) [right = 3 cm of s5] {$(2,7)$};
		\node[block] (s7) [right = 3 cm  of s6] {$(7,7)$};
		\node[block] (s8) [right = 4 cm of s7] {$(3,4)$};  
		\node[block] (s9) [right = 5 cm  of s8] {$(4,4)$}; 
		\node[block] (s10) [right = 2 cm  of s9]  {$(1,1)$}; 
		\node[block] (s11) [right = 2 cm  of s10]  {$(6,1)$}; 
		\node[block] (s12) [right = 2 cm  of s11] {$(2,1)$}; 
		\node[double,block,shade] (s13) [below left = 2 cm and 2 cm of s4] {$(10,7),a$};
		\node[double,block,shade] (s14) [right = 2cm of s13] {$(10,7),c$};
		\node[double,block] (s15) [right = 1.8cm of s14] {$(5,7),a$};
		\node[double,block,shade] (s16) [right = 1.8cm of s15] {$(5,7),c$};
		\node[double,block,shade] (s17) [right = 1.8cm of s16] {$(2,7),a$};
		\node[double,block] (s18) [right = 1.8cm of s17] {$(2,7),c$};
		\node[double,block,shade] (s19) [right = 1.8cm of s18] {$(7,7),a$};
		\node[double,block,shade] (s20) [right = 1.8cm of s19] {$(7,7),c$};
		
		\node[double,block,shade] (s21) [right = 1.8cm of s20] {$(3,4),a$};
		\node[double,block] (s22) [right = 1.8cm of s21] {$(3,4),b$};
		\node[double,block] (s23) [right = 1.8cm of s22] {$(4,4),b$};
		\node[double,block] (s24) [right = 1.8cm of s23] {$(4,4),a$};
		
		\node[double,block] (s25) [right = 2cm of s24] {$(1,1),b$};
		\node[double,block] (s26) [right = 2cm of s25] {$(6,1),b$};
		\node[double,block,shade] (s27) [right = 2cm of s26] {$(2,1),b$};
		
		\node[block] (s28) [below = of s15] {$(6,8)$}; 
		\node[block] (s29) [below = of s18] {$(3,9)$}; 
		\node[block] (s30) [below = of s22] {$(10,5)$};
		\node[block] (s31) [below = of s23] {$(5,5)$};
		\node[block] (s32) [below = of s24] {$(2,2)$};  
		\node[block] (s33) [below = of s25] {$(10,2)$}; 
		
		\node[double,block] (s34) [below = of s28] {$(6,8),b$};
		\node[double,block] (s35) [below = of s29] {$(3,9),b$};
		\node[double,block,shade] (s36) [left = of s30] {$(10,5),a$};
		\node[double,block] (s37) [below = of s31] {$(5,5),a$};
		\node[double,block] (s38) [below = of s32] {$(2,2),c$};
		\node[double,block,shade] (s39) [right = of s33] {$(10,2),c$};
		
		\node[block] (s40) [left = of s37] {$(6,6)$};
		\node[block] (s41) [right = of s38] {$(3,3)$}; 
		
		\node[double,block] (s42) [left = of s40] {$(6,6),b$};
		\node[double,block] (s43) [below = 1.5 cm of s41] {$(3,3),b$}; 
		
		\node[block] (s44) [below = 1.5 cm of s42] {$(10,10)$}; 		
		
		\path[->]
		(s0) edge node [sloped,above] {$b,v_3$} (s1)  
		(s0) edge node {$c,v_2$} (s2)  
		(s0) edge node {$a,v_1$} (s3)
		
		(s1) edge node [sloped,above] {$abc,cac,cba$} (s4)  
		(s1) edge node {$c$} (s5)  
		(s1) edge node {$a$} (s6) 
		(s1) edge node {$\epsilon$} (s7)  
		
		(s2) edge node [sloped,above] {$ab,ca$} (s8)  
		(s2) edge node {$\epsilon$} (s9) 
		
		(s3) edge node [sloped,above] {$\epsilon$} (s10)  
		(s3) edge node  {$cb$} (s11)  
		(s3) edge node {$c$} (s12)
		
		(s4) edge node [below,left]{$a,v_6$} (s13)  
		(s4) edge node {$c,v_7$} (s14)
		
		(s5) edge node [below,left] {$a,v_6$} (s15) 
		(s5) edge node {$c,v_7$} (s16)
		
		(s6) edge node [below,left] {$a,v_6$} (s17)  
		(s6) edge node {$c,v_7$} (s18)
		
		(s7) edge node [below,left] {$a,v_6$} (s19)  
		(s7) edge node {$c,v_7$} (s20)
		
		(s8) edge node {$a,v_4$} (s21) 
		(s8) edge node {$b,v_5$} (s22)  
		
		(s9) edge node [below,left] {$b,v_5$} (s23)  
		(s9) edge node {$a,v_4$} (s24) 
		
		(s10) edge node {$b$} (s25)  
		
		(s11) edge node {$b$} (s26)  
		
		(s12) edge node {$b$} (s27) 
		
		(s15) edge node  {$\epsilon$} (s28) 
		
		(s18) edge node  {$\epsilon$} (s29)  
		
		(s22) edge node [left] {$\epsilon$} (s30) 
		
		(s23) edge node [sloped,above] {$ac,ba$} (s30) 
		(s23) edge node {$\epsilon$} (s31)  
		
		(s24) edge node {$\epsilon$} (s32) 
		
		(s25) edge node {$\epsilon$} (s32)  
		(s25) edge node {$bc$} (s33)  
		
		(s26) edge node {$\epsilon$} (s33) 
		
		(s28) edge node {$b$} (s34)  
		
		(s29) edge node {$b$} (s35) 
		
		(s30) edge node {$a$} (s36) 
		
		(s31) edge node {$a$} (s37)  
		
		(s32) edge node {$c$} (s38) 
		
		(s33) edge node {$c$} (s39) 
		
		(s34) edge node {$\epsilon$} (s44) 
		
		(s35) edge node {$\epsilon$} (s44)  
		
		(s37) edge node {$\epsilon$} (s40) 
		
		(s38) edge node {$\epsilon$} (s41)  
		
		(s40) edge node {$b$} (s42)
		
		(s41) edge node {$b$} (s43) 
		
		(s42) edge node [left]{$\epsilon$} (s44) 
		
		(s43) edge node {$\epsilon$} (s44)
		
		; 
		\end{tikzpicture} 
		\caption{The obtained PMDP model $\mathcal{M}$, the shaded states are blocking states belong to the sink set $\sqcup$}
		\label{fig:PMDP}
	\end{figure*}

	Take the communication network example as shown in Figure \ref{fig:PSDES}, suppose the communication node $5$ has very limited power and computation capability and we would like to avoid using it too often. Then we could define $D=\{5\}$. Also, we assume that there is an intruder eavesdropping the  transmitted message to determine if the message has reached the communication node $8$ or $9$. However, subjected to the  bandwidth and decoding constraints, the intruder could only partially observe the transmitted message and partially decode the message to know the message type. We would like to design the routing probabilities and also the insertion function, such that the opacity level is no less than $\gamma$ and the task specification to avoid using node $5$ with probability larger than $1-\lambda$ can be satisfied.

	\section{Main Results}\label{sec:main_results}
	\subsection{Obtaining PMDP for Insertion Function Synthesis}

	In this paper, we require the following assumption to guarantee that the structure of the underlying graph of the PSDES $PH$ does not change. 
	
	\begin{assumption}\label{asp:topology}
		For a parametric PSDES $PH=(Q,\Sigma,P,\pi_0,V)$, unless $P_v(q,\sigma,q')\in{0,1}$ for any evaluation, it must hold that
		$
		0<P_v(q,\sigma,q')<1
		$
		$\forall v\in Val^v,\forall q,q'\in Q,\forall \sigma\in \Sigma$
		where $v$ is an valid evaluation of the parameter vector $V$.
	\end{assumption}

	Given an SDES $H$, an MDP $\mathcal{M}$ can be constructed to show all the possible insertion functions or strategies in MDP's term \cite{wu2016synthesis} with the following assumption to guarantee that the resulting $\mathcal{M}$ has a finite state space. 
	
	\begin{assumption}\label{asp:finite}
		Every observable string in $H$ is of a finite length. 
	\end{assumption}
	
	In our case, where the model is a PSDES $PH=(Q,\Sigma,P,\pi_0,V)$, with Assumption \ref{asp:topology} and \ref{asp:finite},  we could follow the same algorithm to obtain a PMDP $\mathcal{M}=(S,\hat{s},A,T,V)$. Note that this PMDP has a unique initial state $\hat{s}$. 
	
	This PMDP $\mathcal{M}=(S,\hat{s},A,T)$ can be seen as a game between the system $PH$ and the insertion mechanism. The states $S$ can be divided into two disjoint sets, namely the system states $S_s$ and insertion states $S_i$. As shown in Figure \ref{fig:PMDP}, the states in single line blocks are system states $S_s=E_i\times E_s\times \pi$, and the states in double line blocks are insertion states $S_i=E_i\times E_s\times\Sigma_o\times\pi$, where $E_i$ is the state estimate of the intruder after the observation mask $O$ and inserted events, which could be wrong since it could be fooled by the inserted events, $E_s$ is the system's state estimate after observation mask $O$, which is always correct because it is aware of what events have been inserted. $\Sigma_o$ denotes the recently observed event from $PH$. $\pi$ denotes the state distribution of $PH$ based on the observed strings without insertions so far. The details of obtaining this PMDP, including constructing the states, computing the belief state $\pi$ and transition probabilities can be found in \cite{wu2016synthesis}.

	Take the PSDES as shown in Figure \ref{fig:PSDES} as an example, the resulting $\mathcal{M}$ is illustrated in Figure \ref{fig:PMDP}. We didn't show $\pi$ in Figure \ref{fig:PMDP} due to the space limitation and $\pi$ is not relevant to our further development. The initial position is $s_0 = ((0,0),\pi(0)=1)$. Suppose the event $b$ in $H$ happens, then the next state which is an insertion state $s'$ will be $s' = ((0,0),b,\pi(7)=1)$ and the transition probability is $v_3$. The estimates of the intruder and the system don't change since the insertion has not been decided yet, and will be handled at the insertion state $s'$. As can be observed in Figure \ref{fig:PSDES}, there are multiple insertion choices that are available at $s'$, such as inserting $abc,cac,cba,c,a$ or the empty string $\epsilon$. For example, if we choose to insert string $abc$ in front of $b$, we will end up in the system state $s'' = ((10,7),\pi(7)=1)$ with probability $1$. That is, the intruder believes that the current state is at $10$, but the system $PH$ is actually at $7$. However, it can be seen from $s''$ that what ever happens in $PH$ later, either $a$ or $c$, we could not find a valid insertion. That is, the insertion strategy gets blocked since there is no valid insertion action available. From Figure \ref{fig:PMDP}, we denote $\sqcup$ as the set of all the shaded states  that are blocking and should be avoided for opacity enforcement. The desired final state is $(10,10),\pi(10)=1$. Once we reach this state, we are done since both the system and the intruder does not expect any new events and the opacity has been preserved along the way.
	
	\begin{remark}
		Note that in this PMDP $\mathcal{M}$, $A=\Sigma_o^*\cup\perp$. For a system state $s\in S_s$, there is only one dummy action $\perp$ defined and with certain probability, $s$ will transit to some  $s'\in S_i$. Therefore, the labels $\sigma\in\Sigma_o$ on the transition from a system state to an insertion state \emph{does not} mean an action in PMDP, as each action in the MDP should incur a distribution that sums up to 1. It simply illustrates the event that has just  happened in the PSDES $PH$. On the other hand, the labels $\omega\in\Sigma_o^*$ on the transition from an insertion state to a system state \emph{means an action} in the PMDP, which refers to the string $\omega$ to be inserted before the recently observed system event. Then after the insertion, with probability $1$, the insertion state transits to a system state. 
	\end{remark}
	It can be seen that the transition probability from an insertion state to a system state on an insertion action $\omega\in\Sigma_o^*$ will always be $1$. And the transition probability from a system state $s$ to an insertion state $s'$, given observed history string $\omega=\sigma_0\sigma_1...\sigma_{k-1}$ and the most recent output event $\sigma_{k}$ can be computed as
	\begin{equation}\label{eqn:T}
		T(s,\perp,s') = ||\pi_0\prod_{i=0}^{k}P_{uo}^*P_{\sigma_i}||_1
	\end{equation}
	where we define $P_{uo}=\sum_{\sigma\Sigma_{uo}}P_\sigma$ in the PSDES $PH=(Q,\Sigma,P,\pi_0,V)$, so $P_{uo}(i,j)$ denotes the probability to transit from $q_i$ to $q_j$ under some unobservable event. Since $PH$ may be partially observable, an arbitrary number of unobservable events could happen between any two observable events. Then $P_{uo}^*=\sum_{i=0}^\infty P_{uo}^i$, where $P_{uo}^*(i,j)$ denote the probability to transit from $q_i$ to $q_j$ under any string of unobservable events. $P_{uo}$ converges to $(I_n-P_{uo})^{-1}$ if $\sum_{q'\in Q}\sum_{\sigma\in\Sigma_{uo}}P(q,\sigma,q')<1$, as assumed in this paper.
	
	To make sure that our PMDP $\mathcal{M}$ satisfies the requirement that all its transition probabilities belong to the posynomial function class. We need the following assumption.
	
	\begin{assumption}\label{asp:Puo}
		Either $P_{uo}$ is a  constant matrix or there exists a finite integer $K$ such that $P_{uo}^k=0,\forall k> K$.
	\end{assumption}
	
	\begin{lemma}
		With Assumption 1, the MDP $\mathcal{M}$ obtained from $PH=(Q,\Sigma,P,\pi_0,V)$ is a parametric MDP with all its transition probabilities in the posynomial function class.
	\end{lemma}
	\begin{proof}
		If every element in a vector/matrix is either a posynomial function or $0$, we call it a quasi-posynomial vector/matrix. Denote $|Q|=n<\infty$, from the definition of posynomial functions in (\ref{eqn:pos}), it can be seen that posynomials are closed under addition and multiplication, and so are the quasi-posynomial matrices. $P_{uo}=\sum_{\sigma\Sigma_{uo}}P_\sigma$ is a quasi-posynomial matrix since $P_\sigma$ is a quasi-posynomial matrix and $\Sigma_{uo}$ is a finite set. Then with Assumption \ref{asp:Puo}, if $P_{uo}$ is a constant quasi-posynomial matrix, $P^*_{uo}$ will converge to $(I_n-P_{uo})^{-1}$, which is also guaranteed to be a quasi-posynomial matrix . Otherwise if $P_{uo}^k=0,\forall k> K$, $P_{uo}^*=\sum_{i=0}^K P_{uo}^i$, and, thus, $P^*_{uo}$ is still a quasi-posynomial matrix.
		
		From (\ref{eqn:T}) where the transition probability is computed, $\pi_0$ is a constant quasi-posynomial vector with $n$ elements. With Assumption \ref{asp:Puo}, $P_{uo}^*$ is guaranteed to be a $n\times n$ quasi-posynomial matrix, $P_{\sigma_i}$ is a quasi-posynomial matrix. Since $k$ is guaranteed to be finite in (\ref{eqn:T}) from Assumption \ref{asp:finite}, $\pi_k=\pi_0\prod_{i=0}^{k}P_{uo}^*P_{\sigma_i}$ is a quasi-posynomial vector, where $\pi_k(s)$ denotes the probability of landing in the state $s$ after observing $\omega=\sigma_0\sigma_1...\sigma_{k}$. Since $\omega$ is guaranteed to be feasible to happen, $\pi_k$ is not possible to be an all-zero vector. Then the sum of all elements in $\pi_k$ will be $||\pi_k||_1 = ||\pi_0\prod_{i=0}^{k}P_{uo}^*P_{\sigma_i}||_1=T(s,\perp,s')$, is a posynomial function.
	\end{proof}

	\subsection{Insertion Function and Parameter Co-synthesis}
	Now that we have a PSDES $PH=(Q,\Sigma,P,\pi_0,V)$, where $V=\{v_1,...,v_m\}$ and its PMDP $\mathcal{M}=(S,\hat{s},A,T,V)$ that encodes all the possible insertions as its actions. To solve Problem 1, we need to find a valuation $v$ and a insertion strategy $\mu(s,\alpha)$ that inserts the string $\alpha\in A=\Sigma_o^*$ at an insertion state $s\in S_i$, to guarantee that both $P^I_{CSO}\geq\gamma$ and $\phi_t = P_{\leq \lambda}(\diamondsuit D)$ are satisfied. 
	
	To make our problem more meaningful, we assume that there does not exist an insertion strategy to enforce the opacity with probability $1$. Otherwise, we could simply use this insertion function and then synthesize parameters to enforcement the task specification separately. Instead, we would like to solve the insertion function and parameter synthesis when they are coupled. Inspired by \cite{cubuktepe2017sequential}, the solution of this parameter and strategy co-synthesis problem can be converted to a nonlinear program (NLP)  (geometric program (GP), to be more specific) as follows.
	\begin{align}
		&\text{minimize } F = \sum_{i=1}^{m}\frac{1}{v_i}+\sum_{s\in S,\alpha\in A(s)} \frac{1}{\mu(s,\alpha)}\label{eqn:f}\\&\text{ subject to} \\
		& \frac{p^o_{\hat{s}}}{1-\gamma}\leq 1\label{eqn:cso} \\
		& \frac{p^t_{\hat{q}}}{\lambda}\leq 1\label{eqn:task} \\
		&\forall s\in S, \sum_{\alpha\in A(s)} \mu(s,\alpha)\leq 1 \label{eqn:scheduler}\\
		&\forall s\in S, \forall \alpha\in A(s), \text{  } \mu(s,\alpha)\leq 1 \label{eqn:scheduler1}\\ 
		&\forall s\in S, \forall \alpha\in A(s), \text{  }\sum_{s'\in S}T(s,\alpha,s')\leq 1 \label{eqn:tsas}\\
		&\forall s,s'\in S, \forall \alpha\in A(s), \text{  } T(s,\alpha,s')\leq 1\label{eqn:tsas1}\\
		&\forall s\in\sqcup, p^o_s = 1\label{eqn:op0}
	\end{align}
	\begin{align}
		&\forall s\in S/\sqcup, \frac{\sum_{\alpha\in A(s)}\mu(s,\alpha)\sum_{s'\in S}T(s,\alpha,s')p^o_{s'}}{p^o_s} \leq 1\label{eqn:op} \\
		&\forall q\in Q, \sum_{q'\in Q}\sum_{\sigma\in\Sigma}P(q,\sigma,q')\leq 1 \label{eqn:psas}\\
		&\forall q\in D, p^t_q = 1 \label{eqn:NLP0}\\
		&\forall q\in Q\backslash D, \frac{\sum_{\sigma\in\Sigma}\sum_{q'\in Q}P(q,\sigma,q')p^t_{q'}}{p^t_q} \leq 1  \label{eqn:NLP}
	\end{align}
	
	Equation (\ref{eqn:cso}) and (\ref{eqn:task}) encode the CSO requirement and the task specification respectively. Intuitively, from (\ref{eqn:op0}) and (\ref{eqn:op}), it can be observed that $p^o_{s}$ denotes the upper bound of the probability to reach the sink set $\sqcup$ from a state $s$ in the PMDP $\mathcal{M}$ and thus $P_{CSO}\geq 1-p^o_{\hat{s}}\geq\gamma$ where $\hat{s}$ is the initial state of the PMDP $\mathcal{M}$.  From (\ref{eqn:NLP0}) and (\ref{eqn:NLP}), it can be seen that $p^t_{s}$ denotes the upper bound of the probability to reach the undesired set $D$, so from the specification $\phi_t$, we require $p^t_{\hat{q}}\leq\lambda$ where $\hat{q}$ is the initial state of the PSDES $PH$. Equation (\ref{eqn:scheduler}) denotes the requirement for the scheduler of the PMDP, observe that $F$ in (\ref{eqn:f}) is monotonic with regard to $\mu(s,\alpha)$, therefore the optimal solution from this NLP will achieve the equality, which satisfies the requirement that the probability of the scheduler's choice at each state should sum up to $1$. Equation (\ref{eqn:scheduler1}) and (\ref{eqn:tsas1}) requires that the probabilities should be bounded by $1$. As to (\ref{eqn:tsas}), $F$ is monotonic with respect to $V$, if $T(s,\alpha,s')$ is a posynomial of $V$, it will also be monotonic with respect to $V$. Then the equality will be achieved, which satisfies the requirement that the action $\alpha$ induces a distribution that sums up to $1$. The same argument applies to (\ref{eqn:psas}).
	
	\begin{thm}
		With the encoding (\ref{eqn:f}) - (\ref{eqn:NLP}), if the solution finds a well-defined scheduler and valuation $v$,  this solution then solves the Problem $1$ and respects Assumption \ref{asp:wu} and \ref{asp:topology}.
	\end{thm}
	
	\begin{proof}
		From (\ref{eqn:f}), it can be seen that to minimize $F$, the parameters $v_i\in V$ are not going to be zero, which implies that Assumption \ref{asp:topology} holds. Because of the monotonicity of $F$ with respect to $V$ and (\ref{eqn:psas}), $\sum_{q'\in Q}\sum_{\sigma\in\Sigma}P(q,\sigma,q')\in\{0,1\}$ in Assumption \ref{asp:wu} is assured.  
		
		Since $v_i\neq 0, \forall i\in[1,m]$, we have $P(q,\sigma,q')>0$ for any $q,q',\sigma\in\Sigma_o$, therefore $\sum_{q'\in Q}\sum_{\sigma\in\Sigma_{uo}}P(q,\sigma,q')<1$ is also assured. Thus, Assumption \ref{asp:wu} is satisfied. 
		
		As for the $P_{CSO}$ requirement, since $p^o_{\hat{s}}\geq 1-P_{CSO}$ and $p^o_{\hat{s}}\leq 1-\gamma$ as required in (\ref{eqn:cso}), it implies that $P_{CSO}\geq1-p^o_{\hat{s}}\geq\gamma$. Similarly for $\phi_t$ requirement, $\lambda\geq p^t_{\hat{q}}\geq P(\diamondsuit D)$.
	\end{proof}
	\begin{remark}
		Note that in this NLP encoding, we assume $PH$ has a unique initial state, while in our previous definition of PSDES, our initial condition is a distribution on the states. This can be easily converted to the unique initial state case by adding a dummy initial state $\hat{q}$ and the transition probability from $\hat{q}$ is according to the initial distribution $\pi_0$.
		
		Furthermore, in Problem $1$, we are looking for a deterministic insertion function $I$. But from the encoding, we get a probabilistic insertion function where the probability of inserting a string $\alpha$ is determined by $\mu(s,\alpha)$. From $F$ we know that $\mu(s,\alpha)>0,\forall s\in S, \alpha\in A(s) $, so if there are multiple insertion choices, our insertion strategy may choose any one of them with non-zero probability. But nevertheless, our resulting insertion strategy still satisfies the CSO requirement.  
		
		Also, note that this encoding is not complete. If we fail to  find a valid solution, it doesn't mean that a valid solution doesn't exist.
	\end{remark}
	
	\subsection{An Illustrative Example}
	Let's return to our motivating example as shown in Figure \ref{fig:PSDES} and its PMDP in Figure \ref{fig:PMDP}. The secret state are $8$ and $9$, and the state we would like to avoid is $5$. The opacity and reachability constraints are $P_{CSO}\leq\gamma=0.15$ and $\phi_t=P_{\leq\lambda}(\diamondsuit D)$ where $\lambda=0.3$. We encode our problem following (\ref{eqn:f})-(\ref{eqn:NLP}) and solve it using the optimization solver GGPLAB \cite{mutapcic2006ggplab}.
	The resulting parameters are $v_1=0.3501,v_2=0.3501,v_3=0.2998,v_4=0.5,v_5=0.5,v_6=0.5,v_7=0.5$. And the resulting insertion strategy is shown in Figure \ref{fig:solvedPMDP}. On the transition arrow from the insertion function station to a system state, the number after the inserted string denotes the probability that the insertion function chooses to select this insertion action. For example, from the initial state, event $a$ has probability $v_1$ to happen, and then the system transits to $((0,0),a)$ where, with probability $1-2*10^{-5}$, nothing is inserted and then the MDP transits to $(1,1)$ with probability $1$, which implies that both the real and intruder's state estimation are $1$. With our synthesized parameters and the insertion strategy, $P^I_{CSO}=0.15$ and $P(\diamondsuit D)=0.2507$. Note that in this particular example, for instance, the insertions $abc,cac,cba$ all make the state transits from $((0,0),b)$ to $(10,7)$ with probability $1$. Our synthesized strategy actually  assigns a total probability of $10^{-5}$ to choose among $abc,cac,cab$ while not specifying the exact probability to choose each individual string. 
	
	\begin{figure*}[ht]
		\centering
		\begin{tikzpicture}[shorten >=1pt,node distance=2cm,on grid,auto, bend angle=20, thick,scale=0.55, every node/.style={transform shape}] 
		\node[block,initial above] (s0)   {$(0,0)$}; 
		\node[double,block] (s1) [below left = 1.5 cm and 10 cm of s0] {$(0,0),b$}; 	
		\node[double,block] (s2) [below = 1.5 cm of s0]  {$(0,0),c$}; 
		\node[double,block] (s3) [below right= 1.5 cm and 8 cm of s0] {$(0,0),a$}; 
		
		\node[block] (s4) [below left= 4 cm and 10 cm of s1] {$(10,7)$}; 
		\node[block] (s5) [right = 4 cm  of s4] {$(5,7)$}; 
		\node[block] (s6) [right = 3 cm of s5] {$(2,7)$};
		\node[block] (s7) [right = 3 cm  of s6] {$(7,7)$};
		\node[block] (s8) [right = 4 cm of s7] {$(3,4)$};  
		\node[block] (s9) [right = 5 cm  of s8] {$(4,4)$}; 
		\node[block] (s10) [right = 3 cm  of s9]  {$(1,1)$}; 
		\node[block] (s11) [right = 3 cm  of s10]  {$(6,1)$}; 
		\node[block] (s12) [right = 3 cm  of s11] {$(2,1)$}; 
		\node[double,block,shade] (s13) [below left = 2 cm and 2 cm of s4] {$(10,7),a$};
		\node[double,block,shade] (s14) [right = 2cm of s13] {$(10,7),c$};
		\node[double,block] (s15) [right = 1.8cm of s14] {$(5,7),a$};
		\node[double,block,shade] (s16) [right = 1.8cm of s15] {$(5,7),c$};
		\node[double,block,shade] (s17) [right = 1.8cm of s16] {$(2,7),a$};
		\node[double,block] (s18) [right = 1.8cm of s17] {$(2,7),c$};
		\node[double,block,shade] (s19) [right = 1.8cm of s18] {$(7,7),a$};
		\node[double,block,shade] (s20) [right = 1.8cm of s19] {$(7,7),c$};
		
		\node[double,block,shade] (s21) [right = 1.8cm of s20] {$(3,4),a$};
		\node[double,block] (s22) [right = 1.8cm of s21] {$(3,4),b$};
		\node[double,block] (s23) [right = 1.8cm of s22] {$(4,4),b$};
		\node[double,block] (s24) [right = 1.8cm of s23] {$(4,4),a$};
		
		\node[double,block] (s25) [right = 3cm of s24] {$(1,1),b$};
		\node[double,block] (s26) [right = 3cm of s25] {$(6,1),b$};
		\node[double,block,shade] (s27) [right = 3cm of s26] {$(2,1),b$};
		
		\node[block] (s28) [below = 4 cm of s15] {$(6,8)$}; 
		\node[block] (s29) [below = 4cm of s18] {$(3,9)$}; 
		\node[block] (s30) [below = 4cm of s22] {$(10,5)$};
		\node[block] (s31) [below = 4cm of s23] {$(5,5)$};
		\node[block] (s32) [below = 4cm of s24] {$(2,2)$};  
		\node[block] (s33) [below = 4cm of s25] {$(10,2)$}; 
		
		\node[double,block] (s34) [below = of s28] {$(6,8),b$};
		\node[double,block] (s35) [below = of s29] {$(3,9),b$};
		\node[double,block,shade] (s36) [left= of s30] {$(10,5),a$};
		\node[double,block] (s37) [below = of s31] {$(5,5),a$};
		\node[double,block] (s38) [below = of s32] {$(2,2),c$};
		\node[double,block,shade] (s39) [right = of s33] {$(10,2),c$};
		
		\node[block] (s40) [left = of s37] {$(6,6)$};
		\node[block] (s41) [below = of s38] {$(3,3)$}; 
		
		\node[double,block] (s42) [left = of s40] {$(6,6),b$};
		\node[double,block] (s43) [left = of s41] {$(3,3),b$}; 
		
		\node[block] (s44) [below = of s42] {$(10,10)$}; 		
		
		\path[->]
		(s0) edge node [sloped,above] {$b,v_3$} (s1)  
		(s0) edge node {$c,v_2$} (s2)  
		(s0) edge node {$a,v_1$} (s3)
		
		(s1) edge node [sloped,above] {$(abc,cac,cba),10^{-5}$} (s4)  
		(s1) edge node [sloped,above]{$c,0.5-10^{-5}$} (s5)  
		(s1) edge node [sloped,above]{$a,0.5-10^{-5}$} (s6) 
		(s1) edge node [sloped,above]{$\epsilon,10^{-5}$} (s7)  
		
		(s2) edge node [sloped,above] {$(ab,ca),10^{-5}$} (s8)  
		(s2) edge node {$\epsilon,1-10^{-5}$} (s9) 
		
		(s3) edge node [sloped,above] {$\epsilon,1-2*10^{-5}$} (s10)  
		(s3) edge node [sloped,above] {$cb,10^{-5}$} (s11)  
		(s3) edge node [sloped,above] {$c,10^{-5}$} (s12)
		
		(s4) edge node [below,left]{$a,v_6$} (s13)  
		(s4) edge node {$c,v_7$} (s14)
		
		(s5) edge node [below,left] {$a,v_6$} (s15) 
		(s5) edge node {$c,v_7$} (s16)
		
		(s6) edge node [below,left] {$a,v_6$} (s17)  
		(s6) edge node {$c,v_7$} (s18)
		
		(s7) edge node [below,left] {$a,v_6$} (s19)  
		(s7) edge node {$c,v_7$} (s20)
		
		(s8) edge node {$a,v_4$} (s21) 
		(s8) edge node {$b,v_5$} (s22)  
		
		(s9) edge node [below,left] {$b,v_5$} (s23)  
		(s9) edge node {$a,v_4$} (s24) 
		
		(s10) edge node {$b$} (s25)  
		
		(s11) edge node {$b$} (s26)  
		
		(s12) edge node {$b$} (s27) 
		
		(s15) edge node  {$\epsilon$} (s28) 
		
		(s18) edge node  {$\epsilon$} (s29)  
		
		(s22) edge node [left] {$\epsilon$} (s30) 
		
		(s23) edge node [sloped,above] {$(ac,ba),10^{-5}$} (s30) 
		(s23) edge node [sloped,above] {$\epsilon,1-10^{-5}$} (s31)  
		
		(s24) edge node {$\epsilon$} (s32) 
		
		(s25) edge node [above,sloped] {$\epsilon,1-10^{-5}$} (s32)  
		(s25) edge node [above,sloped] {$bc,10^{-5}$} (s33)  
		
		(s26) edge node {$\epsilon$} (s33) 
		
		(s28) edge node {$b$} (s34)  
		
		(s29) edge node {$b$} (s35) 
		
		(s30) edge node {$a$} (s36) 
		
		(s31) edge node {$a$} (s37)  
		
		(s32) edge node {$c$} (s38) 
		
		(s33) edge node {$c$} (s39) 
		
		(s34) edge node {$\epsilon$} (s44) 
		
		(s35) edge node {$\epsilon$} (s44)  
		
		(s37) edge node {$\epsilon$} (s40) 
		
		(s38) edge node {$\epsilon$} (s41)  
		
		(s40) edge node {$b$} (s42)
		
		(s41) edge node {$b$} (s43) 
		
		(s42) edge node [left]{$\epsilon$} (s44) 
		
		(s43) edge node {$\epsilon$} (s44)
		
		; 
		\end{tikzpicture} 
		\caption{$\mathcal{M}$ with synthesized parameters and insertion strategies, where $v_1=0.3501,v_2=0.3501,v_3=0.2998,v_4=0.5,v_5=0.5,v_6=0.5,v_7=0.5$}
		\label{fig:solvedPMDP}
	\end{figure*}
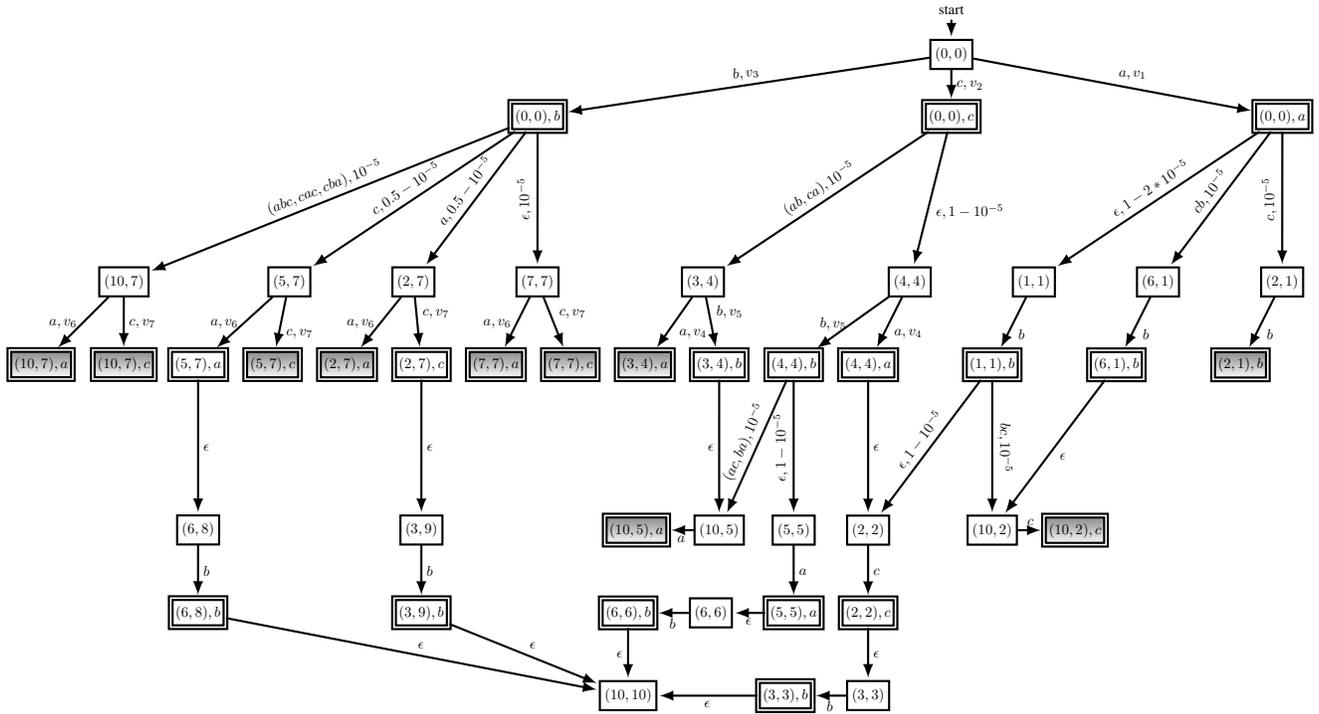
	
	\section{Conclusion}\label{sec:conclusion}
	In this paper, we solved an insertion function and parameter co-synthesis problem on a parametric model, such that both the opacity requirement and task specification can be enforced. The problem was encoded to a nonlinear program and solved. We showed that the solution of this program is a valid one if it respects all the constraints. Future work will consider distributed co-synthesis framework with multiple intruders.
	
	\bibliographystyle{IEEEtran}
	\bibliography{root}
\end{document}